\documentclass[sigconf,balance=false]{acmart}
\acmConference[ICAIF '26 Workshop]{ICAIF 2026 Workshop on Financial AI Security, Privacy and Safety}{November 2026}{Milan, Italy}
\acmYear{2026}
\copyrightyear{2026}
\setcopyright{none}
\acmDOI{}
\acmISBN{}
\usepackage{booktabs}
\usepackage{balance}
\usepackage{amsmath}
\newcommand{\eobi}{\textsc{\textit{Eobi}}}
\newcommand{\obs}{\textit{Observable}}
\newcommand{\obsd}{\textit{Observable+degree}}
\begin{document}

\title{Pricing Cross-Bank Context in Federated Anti-Money Laundering}
\author{Lidia Losavio}
\affiliation{%
  \institution{Universit\`a della Svizzera italiana (USI)}
  \city{Lugano}
  \country{Switzerland}}
\email{lidia.anna.maria.losavio@usi.ch}

\author{Dario Fenoglio}
\affiliation{%
  \institution{Universit\`a della Svizzera italiana (USI)}
  \city{Lugano}
  \country{Switzerland}}
\email{dario.fenoglio@usi.ch}

\author{Martin Gjoreski}
\affiliation{%
  \institution{Universit\`a della Svizzera italiana (USI)}
  \city{Lugano}
  \country{Switzerland}}
\email{martin.gjoreski@usi.ch}

\author{Marc Langheinrich}
\affiliation{%
  \institution{Universit\`a della Svizzera italiana (USI)}
  \city{Lugano}
  \country{Switzerland}}
\email{marc.langheinrich@usi.ch}

\renewcommand{\shortauthors}{Losavio et al.}
\begin{abstract}
Most graph-based Anti-Money Laundering (AML) methods assume access to a pooled transaction graph containing all accounts and payments, which in reality does not exist. Federated learning offers the necessary privacy guarantees but it remains unclear how money laundering patterns spanning separate graphs can effectively be detected. 
We exploit transactions already recorded by both endpoint banks to approximate remote account states locally, enabling multi-layer Graph Neural Network (GNN) computation without additional account-specific exchange. Starting from this baseline, we test whether the remaining context is worth sharing, using a disclosure ladder that adds two transaction counts per remote account or the counterparty's own GNN states. On AMLworld HI-Small, across six GNN backbones and up to 32 silos, the observable baseline yields the largest gain and keeps it as fragmentation shrinks the visible slice; counts add a robust correction, while exact remote states help only in specific regimes. In controlled privacy--utility studies, the noise needed to protect shared states erases their detection value before it buys meaningful privacy, whereas counts stay useful even under strong privacy protection. These findings support a practical rule for banks: use what they already record, and formally protect the little that must move.
\end{abstract}

\begin{CCSXML}
<ccs2012>
 <concept><concept_id>10010147.10010257</concept_id><concept_desc>Computing methodologies~Machine learning</concept_desc><concept_significance>500</concept_significance></concept>
 <concept><concept_id>10002978.10003014</concept_id><concept_desc>Security and privacy~Privacy protections</concept_desc><concept_significance>300</concept_significance></concept>
</ccs2012>
\end{CCSXML}
\ccsdesc[500]{Computing methodologies~Machine learning}
\ccsdesc[300]{Security and privacy~Privacy protections}
\keywords{anti-money laundering, federated graph learning, information disclosure, privacy--utility}
\maketitle

\section{Introduction}
Money laundering conceals illicit funds through sequences of transactions that obscure their origin, with an estimated scale of 2--5\% of global GDP per year \cite{UNODC_ML_estimates, AML_EU}. Individual transfers are rarely suspicious in isolation. Instead, illicit activity emerges through flow patterns such as cycles, fan-in/fan-out structures, and chains of intermediate accounts (Figure~\ref{fig:patterns}) \cite{AMLlitRev}. Graph Neural Networks (GNNs) can jointly exploit transaction attributes and the surrounding flow structure, making them well suited to anti-money laundering (AML) detection. Most graph-based AML methods, however, assume access to a pooled transaction graph containing all accounts and payments \cite{altman2023realistic,egressy2024provably,lin2024fraudgt}. Across financial institutions, each bank observes only part of this network, and data fragmentation is recognized both as an AML vulnerability and as a motivation for privacy-conscious collaborative analytics \cite{fatf2021datapooling,fatf2022partnering}.
\begin{figure}[tb]
  \centering
  \includegraphics[width=\linewidth]{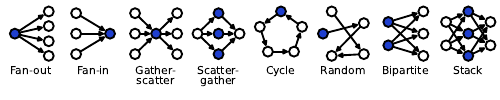}
  \caption{Money laundering patterns \cite{altman2023realistic}. The blue node indicates the main laundering account.}
  \label{fig:patterns}
\end{figure}
Federated learning offers a natural way to train models across institutions by exchanging model updates rather than centralizing their raw data \cite{FLsurvey2026}. Building on this principle, federated graph-learning methods have been developed for graphs partitioned across multiple clients \cite{ambekar2024fedgat,yao2023fedgcn,zhang2021subgraph,lei2023fedcog}, while their application to AML remains underexplored \cite{guo2026cefedgnn}. A further complication is that federated GNN methods are typically tied to how a particular architecture combines graph context. This matters for AML because several AML-tailored solutions that achieve leading performance on centralized benchmarks rely on attention or Transformer-style computation \cite{lin2024fraudgt,losavio2026salt}, which are known to be particularly difficult to federate \cite{yao2023fedgcn,ambekar2024fedgat}. In order to design such collaborative cross-bank AML mechanisms, two more basic questions arise: \emph{(1) what information is genuinely unavailable to each institution; and (2) when does the marginal detection value of sharing it justify the additional privacy exposure?} We study these questions for cross-bank transaction networks.

In contrast to generic federated graph settings, a cross-bank payment is recorded by both endpoint banks with its amount, time, direction, currency, and payment format. Much of the context that federated protocols attempt to recover therefore already sits in each bank's own ledger; what remains genuinely hidden is the counterparty's other activity -- within itself or with third institutions. This motivates \obs{}, a baseline in which each bank approximates a counterparty account's state from the transactions it already records with that account, pooled across all its own accounts. \obs{} uses this approximation wherever the model would otherwise require the remote state, thus providing a no-additional-payload reference (beyond model updates) for measuring how much useful cross-bank context is already locally available and how much value remains in the hidden activity.

We then study a ladder of increasing disclosures: \obs{} defines the zero-additional-payload reference; \obsd{} additionally reveals the remote account's incoming and outgoing transaction counts; and \eobi{} (Exact One-hop Boundary Interface) replaces the local approximation with boundary states computed by the counterparty bank. Such state exchange is established in distributed and federated graph learning \cite{zheng2020distdgl,wan2022pipegcn,guo2026cefedgnn,li2024historical,naman2024optimizing} and serves here as our richer disclosure reference. Our contribution is to measure the incremental value of residual information over the no-additional-payload baseline as the graph becomes increasingly fragmented, and to test whether that value survives formal protection. Importantly, the three interfaces specify what cross-bank information is available rather than how a GNN must aggregate it, making the ladder applicable across GNN architectures, including attention-based models. We evaluate it on AMLworld HI-Small \cite{altman2023realistic} across six GNN backbones under increasing institutional fragmentation, while controlled studies isolate the privacy--utility trade-off of sharing transaction counts and account states.

\textbf{Our key contributions are:} (i) we benchmark attention- and
non-attention-based GNNs tailored to AML under federated graph
fragmentation, across up to 32 silos; (ii) we test three backbone-agnostic cross-bank information interfaces
that isolate the incremental detection utility of increasingly rich
residual context over a zero-additional-payload baseline; and (iii) we quantify how privacy protection changes the utility of residual
degree and state releases under formal privacy guarantees.

\begin{figure*}[t]
\centering
\includegraphics[width=.95\textwidth]{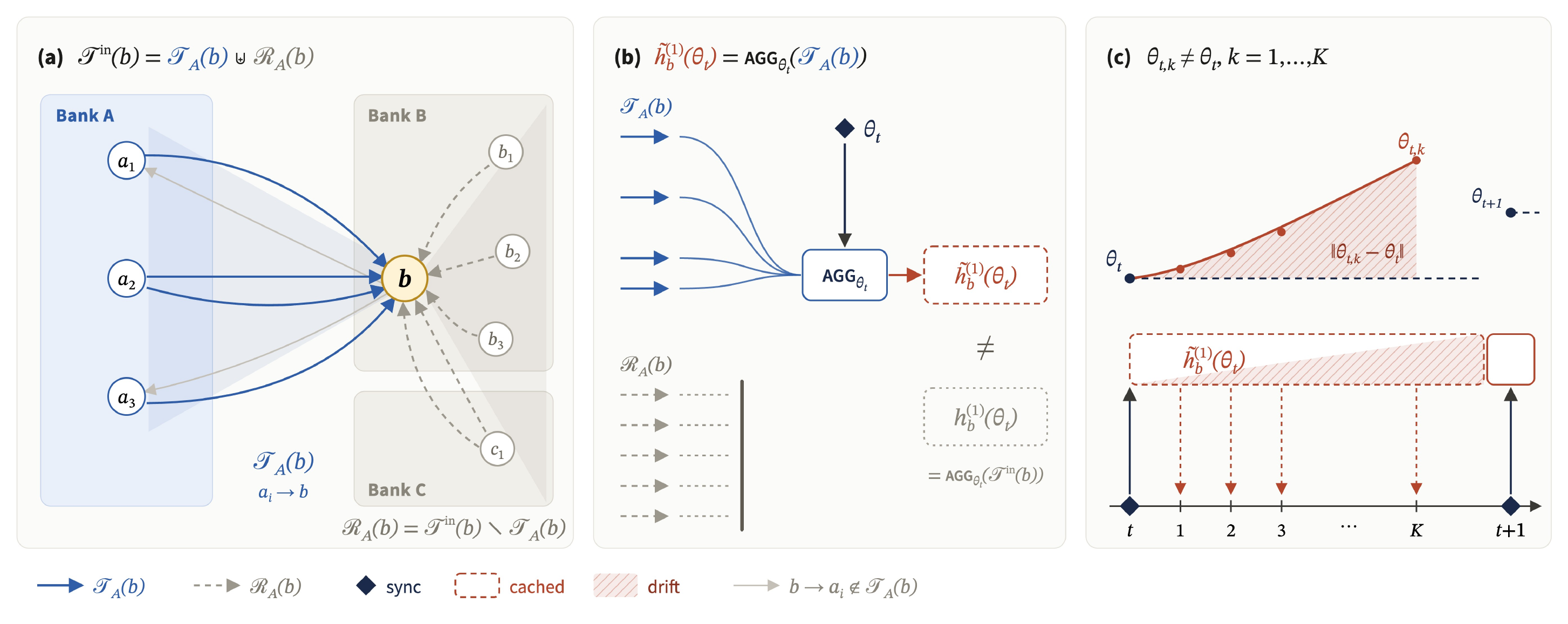}
\caption{\textbf{Observable builds remote context locally.} (a) $\mathcal T_A(b)$ here denotes the visible incoming slice. (b) Bank $A$ uses it to approximate $b$'s first-hop state under $\theta_t$, without exchange. (c) The approximation is cached during local updates and refreshed at synchronization.}
\label{fig:observable}
\end{figure*}

\section{Problem Setting}
\label{sec:setting}
Let $G=(V,E)$ be a directed transaction multigraph of accounts $V$ and
transactions $E$, retaining parallel edges. Each $e=(u,v)\in E$ has attributes
$x_e$ and binary label $y_e\in\{0,1\}$; the task is edge classification.

Let $\mathcal C=\{1,\ldots,M\}$ be the set of $M$ federated clients and let
$o:V\rightarrow\mathcal C$ map each account to its owning client. An edge
$e=(u,v)$ is \emph{internal} if $o(u)=o(v)$ and \emph{cross} otherwise. A
cross edge is replicated in the local graphs of its two endpoint clients for
message passing, but its supervised example is owned only by the sender-side
client $o(u)$. It therefore contributes once to loss and evaluation.
Federated training follows synchronous FedAvg: a coordinator broadcasts the
global model, clients update it on their local graphs for a fixed number of
local epochs, and the returned updates are averaged.

At layer $\ell$, let $h_v^{(\ell)}$ denote the state of account $v$ and let
$\mathcal N(v)$ be the multiset of neighbor--transaction pairs $(u,e)$ used
by the backbone for $v$, so parallel transactions contribute separate
messages. Let $\psi_\ell$ be the edge-conditioned message
function, $\mathrm{AGG}_\ell$ a permutation-invariant aggregator, and
$U_\ell$ the state-update function. A generic message-passing layer is
\begin{equation}
\small
h_v^{(\ell)} = U_\ell\Big(h_v^{(\ell-1)},\
\mathrm{AGG}_{\ell}\!\left(
\left\{\psi_\ell\big(h_v^{(\ell-1)},h_u^{(\ell-1)},x_e\big):
(u,e)\in\mathcal N(v)\right\}\right)\Big).
\label{eq:mp}
\end{equation}
Because laundering typologies are flow patterns over transactions, the
model class of interest is edge-driven message passing: the discriminative
inputs are transaction attributes, while node states accumulate and route
them. We call $G$ \emph{edge-observable} for this base model class when: \textbf{(A1)} every initial state is the same public vector, $h_v^{(0)}=c$, and the first-layer message uses only $c$ and transaction attributes; \textbf{(A2)} for every cross edge $e=(u,v)$, both endpoint clients
know its direction and attributes $x_e$.

\section{Methodology}
\label{sec:method}
Our methodology compares interfaces of increasing information disclosure
while holding the backbone, partition, supervised population, and FedAvg
protocol fixed.

\subsection{Defining Observable and Residual Context}
\label{sec:context}
Consider two distinct clients $A,B\in\mathcal C$ and an account $b$ owned
by $B$, so $o(b)=B$. We call $B$ the \emph{owner} and $A$ the
\emph{observer}; the ordering matters because information may flow from
$B$ to $A$. The complete transaction multiset incident to $b$ is
$\mathcal T(b)=\{e=(u,v)\in E:u=b\ \text{or}\ v=b\}$, and the
\emph{observable slice} of $b$ at client $A$ is
$\mathcal T_A(b)=\{e=(u,v)\in\mathcal T(b):o(u)=A\ \text{or}\ o(v)=A\}$,
namely all transactions between $b$ and \emph{any} account owned by $A$,
not only the account at the opposite endpoint of a target edge.
The \emph{residual context} hidden from $A$ is the multiset difference
$\mathcal R_A(b)=\mathcal T(b)\setminus\mathcal T_A(b)$.
Account $b$ is a directional frontier account from owner $B$ to observer
$A$ when $\mathcal T_A(b)$ is nonempty; the corresponding directional
frontier set is
$\partial_{B\to A}=\{b\in V:o(b)=B,\ \mathcal T_A(b)\ne\varnothing\}$.

The observable slice is only a fraction of $b$'s complete activity. Let
$N_b=|\mathcal T(b)|$ and $n_{A,b}=|\mathcal T_A(b)|$, parallel
transactions included; the \emph{observable coverage} of $b$ at $A$ is
$\rho_{A,b}=n_{A,b}/N_b\in(0,1]$, positive for every frontier account.
Splitting $N_b$ by direction gives the degree vector
$\mathbf d_b=(N_b^{\mathrm{in}},N_b^{\mathrm{out}})$, the payload of
Section~\ref{sec:ladder}.

Assumptions (A1)--(A2) determine who can compute what exactly. The owner
computes $h_b^{(1)}$ exactly from its complete incident view, whereas the
observer, holding only $\mathcal T_A(b)$, can at best approximate it.
Exact deeper states cost a federation-wide exchange: $B$ obtains
$h_b^{(2)}$ exactly only after receiving the first-hop states of its own
frontier accounts \cite{zheng2020distdgl,wan2022pipegcn,wan2022bnsgcn},
and with synchronized parameters and no stochastic operations this reproduces centralized forward values (Appendix~\ref{app:equivalence}). The interfaces of
Section~\ref{sec:ladder} differ precisely in whether $h_b^{(1)}$ is
approximated by the observer or supplied by the owner.

\subsection{Cross-Bank Interfaces}
\label{sec:ladder}
For a frontier account $b\in\partial_{B\to A}$, let $P_{B\to A}(b)$
denote the account-specific payload sent by $B$ to $A$ beyond the model
parameters exchanged by FedAvg. We compare three interfaces
(Figure~\ref{fig:interfaces}, Appendix~\ref{app:protocol}), beginning
with one that requires no such payload.

\textbf{(i) \obs{}.}
\obs{} sends no account-specific payload,
$P^{\mathrm{obs}}_{B\to A}(b)=\varnothing$, and uses only the observable
slice $\mathcal T_A(b)$. At synchronization $t$, observer $A$ evaluates
the first-layer update in Eq.~\eqref{eq:mp} under the refreshed global
parameters $\theta_t$, using the transactions in $\mathcal T_A(b)$
according to the backbone's message directions. This produces
$\widetilde h_{A,b}^{(1)}(\theta_t)$, a local approximation of the
owner-computed state $h_b^{(1)}(\theta_t)$. The approximation pools visible transactions across all relevant
$A$-owned accounts, supplies the remote state for subsequent computation,
and is cached until the next synchronization. Thus \obs{} supports
multi-layer computation without residual transactions or boundary states,
at the cost of missing-context error and staleness.
Figure~\ref{fig:observable} illustrates this construction for incoming
aggregation.

\textbf{(ii) \obsd{}.}
\obsd{} transmits, at each temporal snapshot $s$, the true directed
degree vector,
$P^{\mathrm{deg}}_{B\to A,s}(b)=\mathbf d_{b,s}$.
These are complete transaction counts; the observer can subtract its
own counts to recover the numbers of missing incoming and outgoing
transactions. At the observer, $\mathbf d_{b,s}$ is log-transformed,
standardized, and concatenated to $b$'s public initial state $c$,
augmenting the base initialization in (A1). Message passing still uses only $\mathcal T_A(b)$.
Counts are released at three cumulative snapshots aligned
with the train/validation/test boundaries, using only
transactions recorded up to each snapshot.

\textbf{(iii) \eobi{}.}
\eobi{} releases the learned cut-state record
$P^{\mathrm{bx}}_{B\to A,t}(b)=s_{B\to A,t}(b)$ at round $t$.
This record contains the owner-computed dynamic tensors associated
with $b$ that the model requires across the $B$-to-$A$ cut under global
parameters $\theta_t$. For our two-layer edge classifiers, these
include the first-layer state $h_b^{(1)}$ and, when consumed by the
edge readout, the final endpoint state $h_b^{(2)}$; records are
released in clear.
\eobi{} operates at synchronization. When the server broadcasts
$\theta_t$, each owner evaluates the records of its frontier accounts
from its complete incident view, after the intermediate exchange that
deeper states require (Section~\ref{sec:context}). The observer then
substitutes them for its local approximations. Under (A1)--(A2) the
received $h_b^{(1)}$ equals the centralized one-hop state (computed
without stochastic operations such as dropout), hence the name. With
complete records, the forward pass at the refresh point coincides with
the centralized one; training does not. Received tensors are detached
and cached for one local epoch, growing stale as local weights move, as in historical-embedding schemes \cite{wan2022pipegcn,li2024historical,naman2024optimizing}.

\section{Experimental Setup}
\label{sec:protocol}
We evaluate AMLworld HI-Small under its standard chronological split and transductive protocol \cite{altman2023realistic}, under which the degree snapshots of Section~\ref{sec:ladder} entail no leakage. Banks are assigned to
$M\in\{4,8,16,32\}$ silos, preserving account ownership.
We consider six GNN backbones: GIN, PNA, TransConv, SALT-Trans,
SALT-GAT, and FraudGT
\cite{xu2019gin,corso2020pna,shi2021transconv,losavio2026salt,lin2024fraudgt}.
All federated runs use 100 FedAvg rounds and one local epoch per round.
For each backbone, we reuse the hyperparameters tuned in previous
centralized AML studies
\cite{altman2023realistic,lin2024fraudgt,losavio2026salt},
holding them fixed across the evaluated interfaces.
Model selection maximizes validation F1 score. We report pooled
minority-class test F1 score as mean $\pm$ standard deviation over three
training seeds with a fixed partition. Alongside our three interfaces we include three references:
Local-only, where each silo trains on its own subgraph with no
federation; edge-dropping FedAvg, which averages model updates but
discards cross-silo edges; and centralized training on the pooled graph. We also adapt existing federated graph-learning protocols to compatible backbones: FedGCN to GIN/PNA \cite{yao2023fedgcn} and FedGAT to TransConv
\cite{ambekar2024fedgat}. FedSage+ uses its own GraphSAGE backbone
\cite{zhang2021subgraph}. Additional %
details are provided in Appendix~\ref{app:protocol}.

\section{Results}
\label{sec:benchmark}

\begin{figure*}[!t]
\centering
\includegraphics[width=\textwidth]{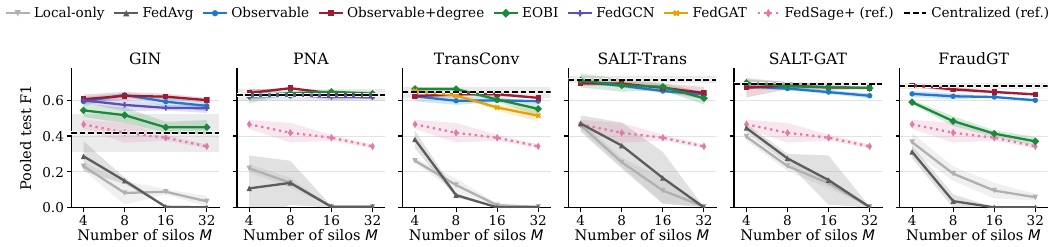}
\caption{HI-Small fragmentation study: pooled test F1 score, mean $\pm$ std over three seeds. Dashed black line and grey band: centralized reference; FedSage+ (pink) repeats across panels as a backbone-independent reference.}
\label{fig:benchmark}
\end{figure*}

Figure~\ref{fig:benchmark} compares the interfaces as fragmentation
increases from $M=4$ to $M=32$ silos. All payloads here are unprotected;
Section~\ref{sec:privacy} asks whether their utility survives protection.

Three main findings emerge. \textbf{First}, the observable slice is a robust proxy for the missing remote state. Its gain over edge-dropping FedAvg (median $+.275$ F1 score at $M=4$) is expected, since it keeps recorded transactions. Its persistence under fragmentation is not. As coverage falls from $.366$ to $.245$ ($M=4$ to $32$), \obs{} loses only $.008$ F1 score for PNA and $.028$ for TransConv, while Local-only and FedAvg approach zero. \textbf{Second}, the cheapest correction is also the most consistently robust overall. At two integers per frontier account, \obsd{} often
matches or exceeds the centralized reference by mean F1 score and remains
competitive under fragmentation. \textbf{Third}, more boundary information does not imply more detection value. \eobi{} helps only in particular backbone--partition
regimes, with no monotone trend in $M$. Among external baselines, FedGCN is stable across silo counts, FedGAT drops from $.653$ to $.514$ F1 score between $M=4$ and $M=32$, and FedSage+'s generated neighbors do not reverse the trend.

\subsection{When Is Residual Context Worth Sharing?}
\label{sec:privacy}
\textbf{Scope and design.} We isolate payload protection in controlled edge-classification tasks. Each example has an observable target transaction and 16 remote boundary accounts, with at most 80 events each. Coverage $\rho$ controls the observable fraction; $\lambda$ controls the target edge's contribution to the label, with smaller $\lambda$ making residual activity more informative. We test selective residuals (e.g., cycles) with TransConv and cardinal residuals (e.g., fan-in/fan-out patterns) with PNA, retaining a distributional negative control. Three datasets and five noise draws per dataset give paired 95\% hierarchical-bootstrap intervals (5,000 resamples). The owner is trusted and the recipient honest-but-curious. Guarantees concern additional releases under fixed observable information and a fixed state encoder; FedAvg model updates are outside scope (Appendix~\ref{app:privacy}). \textbf{Why different payloads help.} Distributional residuals are estimable from a representative slice; cardinal residuals need a complete transaction count (fan-in, structuring); selective residuals depend on one target-matching account whose evidence may be hidden (a cycle-closing witness). The generator varies these mechanisms independently of coverage (Appendix~\ref{app:synthetic}). Table~\ref{tab:clean} (Appendix~\ref{app:synthetic}) first establishes their unprotected utility at $\rho=.1$, $\lambda=.75$: counts recover most of the cardinal gain ($+.033$ vs.\ $+.043$ PR-AUC for \eobi{} over \obs{}), while the selective task benefits chiefly from a complete first-hop state ($+.047$ for \eobi{}, $+.002$ for counts). The privacy studies then test whether these gains survive protection, using the selective state setting at $\lambda=.75$ and the cardinal count setting at $\lambda=.5$.

\begin{figure}[t]
\centering
\includegraphics[width=.8\columnwidth]{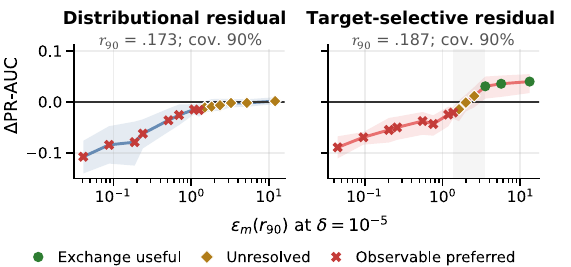}
\caption{\textbf{Protected first-hop states, controlled tasks.} $\Delta$PR-AUC = private \eobi{} $-$ \obs{}, paired 95\% intervals. Horizontal axis: distance-scoped $\epsilon_m(r_{90})$, $\delta=10^{-5}$ (not worst-case $\epsilon_e$); smaller is stronger.}
\Description{Distributional and target-selective residual panels show the utility difference between private state exchange and Observable. Selective exchange is useful only at the weaker end of the distance-scoped privacy range; distributional exchange has no confirmed noisy gain.}
\label{fig:private-state}
\end{figure}

\begin{figure}[t]
\centering
\includegraphics[width=.68\columnwidth]{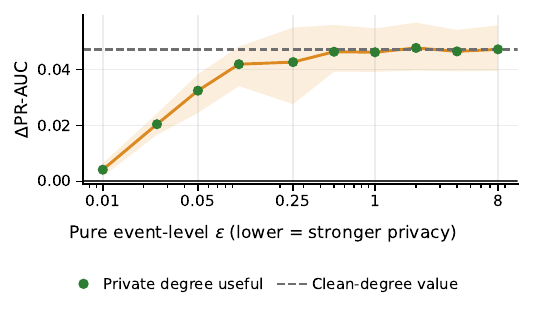}
\caption{\textbf{Protected total counts, cardinal task} (PNA, $\rho=.1$, $\lambda=.5$). $\Delta$PR-AUC relative to \obs{}, paired 95\% intervals; horizontal axis is pure event-level $\epsilon$ per count-vector release.}
\Description{Private count utility remains above Observable for all tested privacy budgets, including epsilon 0.01, and approaches the clean count gain as epsilon increases.}
\label{fig:private-degree}
\end{figure}

\textbf{Boundary states.}
Figure~\ref{fig:private-state} shows that Gaussian noise on clipped
boundary states removes their detection advantage before strong
event-level protection is reached. In the stronger-residual setting ($\lambda=.5$), even an
unprotected gain of $+.147$ PR-AUC is reversed while the
worst-case privacy bound remains weak
($\epsilon_e\approx195$, $\delta=10^{-5}$).
Privacy is accounted per single release, both worst-case ($\epsilon_e$) and scoped to typical state displacements ($\epsilon_m$, Appendix~\ref{app:privacy}), and a separate composition analysis shows that repeated refreshes compound the loss
(Appendix~\ref{app:privacy})
\cite{balle2018analyticGaussian,mironov2017renyi}. These findings
concern the tested state mechanism and do not certify the full
federated training process. \textbf{Degree counts.} In the cardinal PNA task ($\rho=.1$, $\lambda=.5$), the sufficient statistic is one total count $N_b$ per account. Release $N_b+Z$ with two-sided geometric noise, $\Pr[Z=z]\propto\exp(-\epsilon|z|/2)$ \cite{ghosh2009universally}, reusing the same draw across observers. Adding or removing one hidden transaction changes at most two counts by one, so a snapshot release satisfies pure $\epsilon$-DP; three snapshots compose to $3\epsilon$. The clean $+.047$ gain is nearly intact at $\epsilon=.1$ ($+.042$), and remains positive at $\epsilon=.01$ ($+.004$ [$+.002,+.006$]; Figure~\ref{fig:private-degree}). HI-Small's two directed counts have the same sensitivity, but their protected utility is not evaluated here; state and count budgets rest on different adjacency relations and are not directly comparable.

\section{Implications for Banks}
Federated AML is not a uniformly missing-graph problem: observable transactions provide the largest stable gain, degree counts a robust low-cost correction, and richer boundary states conditional value. Counts retain utility under strong formal protection, whereas learned states lose it before their worst-case bound becomes informative. The design rule is therefore: \emph{use what banks already record, and protect only what must move---two integers per shared account when missing cardinality matters---adding richer disclosure only when its value survives protection and refresh.} These results guide payload selection, not end-to-end privacy: AMLworld partitions are controlled abstractions, cached approximations grow stale, and real-bank validation remains open. We deliberately study transaction-driven models; multigraph-specific GNN \cite{egressy2024provably} enhancements remain future work, reverse message passing is compatible, while ego IDs, port encodings, and KYC attributes require additional information and are task-dependent.

\begin{acks}
This work was funded by the ``Resilient Financial Enterprises'' project, supported by the Swiss Innovation Agency (Innosuisse) under grant agreement 119.321 INT-ICT.
\end{acks}

\clearpage
\bibliographystyle{ACM-Reference-Format}
\bibliography{references}
\clearpage
\appendix
\section{Forward Equivalence and Interface Scope}
\label{app:equivalence}
Let $\mathcal N(v)$ be the multiset of neighbor--transaction pairs $(u,e)$ used by the backbone for $v$, so parallel transactions contribute separate messages. A generic edge-conditioned message-passing layer is
\begin{equation}
\begin{split}
h_v^{(\ell)}=U_\ell\Big(h_v^{(\ell-1)},\mathrm{AGG}_\ell\big\{
\psi_\ell(h_v^{(\ell-1)},h_u^{(\ell-1)},x_e):\\
(u,e)\in\mathcal N(v)\big\}\Big).
\end{split}
\end{equation}
The first-hop statement assumes (A1) $h_v^{(0)}=c$ is the same public vector and the first-layer messages use only $c$ and transaction attributes, and (A2) every cross transaction's direction and model-input attributes are available at both endpoint institutions. The neighborhood and any aggregation normalization must be computable from this incident view; additional nonlocal structural inputs are not covered by the statement.

\begin{proposition}[Local first-hop equivalence]
Under (A1)--(A2), with deterministic layers and complete incident inputs, client $o(v)$ computes the same first-layer state $h_v^{(1)}$ for each account it owns as centralized evaluation at the same parameters.
\end{proposition}
\begin{proof}
Client $o(v)$ receives the same incident transaction multiset as centralized evaluation: all internal edges and, by (A2), all cross edges. Assumption (A1) fixes the remaining inputs. The edge-conditioned messages, their aggregation, and the state update therefore coincide.
\end{proof}
This narrow forward equivalence fails with private remote features, missing cross-edge attributes, omitted incident edges, or unavailable nonlocal normalization. It does not imply privacy or an exact backward pass. In particular, it does not equate the observer's $\widetilde h_{A,b}^{(1)}$ with the remote owner's $h_b^{(1)}$. As described in Section~\ref{sec:ladder}, \obs{} constructs the observer's approximation from the observable subset
of the owner's incident inputs. Exact degree vectors are an additional information interface and are not assumed in this zero-payload proposition.

\textbf{Direction in Figure~\ref{fig:observable}.}
The main definitions use the full incident multisets $\mathcal T(b)$ and $\mathcal T_A(b)$. For the figure's incoming channel, define
\[
\mathcal T^{\rm in}(b)=\{(u,b)\in E\},\qquad
\mathcal T_A^{\rm in}(b)=\{(u,b)\in E:o(u)=A\}.
\]
The figure abbreviates $\mathcal T_A^{\rm in}(b)$ and its incoming residual to $\mathcal T_A(b)$ and $\mathcal R_A(b)$. In particular, the pale outgoing edges $b\to a_i$ are observable at $A$ but are not incoming messages to $b$. They are not hidden residual transactions. The illustration retains parallel $a_i\to b$ payments and pools over all $A$-owned accounts connected to $b$. A local approximation can subsequently supply a remote source state along a visible $b\to a_i$ edge, or contextualize a remote endpoint for edge classification.

\textbf{Refresh and exactness.}
At round $t$, the global parameters $\theta_t$ are broadcast before the observer constructs its remote approximations. These states are reused while the local parameters move and rebuilt at the next synchronization. ``One local epoch'' does not imply fixed parameters throughout an epoch: minibatch updates may already create staleness. Figure~\ref{fig:observable} depicts this schedule generically; the plotted drift is schematic.

For an exact synchronized deeper forward computation, owner states must be exchanged in layer order across the federation, together with any model-specific cross-edge tensors. An owner computes $h_b^{(2)}$ exactly only after obtaining the required first-hop states of its own remote neighbors. Its complete incident ledger alone guarantees $h_b^{(1)}$, not all deeper states. With complete records and synchronized deterministic layers, induction then reproduces centralized forward values \cite{zheng2020distdgl,wan2022pipegcn,wan2022bnsgcn}. The benchmark \eobi{} cut record includes $h_b^{(1)}$ and, when consumed by the edge readout, $h_b^{(2)}$. Received states are detached and cached for one local epoch: they are not refreshed after every local optimization step, and no gradient crosses clients. Consequently, neither centralized training equivalence nor a predictive upper bound follows from forward exactness at a refresh point.

\section{HI-Small Protocol and Communication}
\label{app:protocol}
\textbf{Data and partition.}
We preserve the standard chronological AMLworld train/validation/test split \cite{altman2023realistic}. Each bank identifier is assigned to one of $M$ silos and every account inherits its bank's assignment. A silo may therefore represent a group of banks, rather than one real institution. Cross-silo transactions are copied to both endpoint subgraphs for message passing but assigned once to the sender-side silo for supervision and pooled evaluation. At $M=4$, 64.9\% of HI-Small transactions are cross-silo, with 374,791 directional frontier registrations.

\textbf{Backbones and baselines.}
We evaluate GIN \cite{xu2019gin}, PNA \cite{corso2020pna}, TransConv \cite{shi2021transconv}, the statistics-aware attention hybrids SALT-Trans and SALT-GAT \cite{losavio2026salt}, and FraudGT \cite{lin2024fraudgt}. Ego identifiers, port numberings, and reverse message passing are disabled to hold the transaction-driven input regime fixed. Each backbone runs under \obs{}, \obsd{}, and \eobi{}, with Local-only and a centralized reference as anchors and edge-dropping FedAvg \cite{mcmahan2017communication} as the weight-only baseline. Where an external federated protocol admits the backbone, it is adapted: FedGCN \cite{yao2023fedgcn} to GIN/PNA and FedGAT \cite{ambekar2024fedgat} to TransConv. FedSage+ \cite{zhang2021subgraph} carries its own GraphSAGE backbone, so its reference appears across panels.

All federated runs use 100 rounds and one local epoch per round. Backbone hyperparameters and class weights follow the optimal centralized configuration; model selection maximizes validation F1 under an oracle coordinator convention. Reported pooled minority-class test F1 is the mean and standard deviation over three training seeds with a fixed partition. All benchmark comparisons use unprotected payloads; the controlled private releases below are separate experiments.

\begin{figure*}[!t]
\centering
\includegraphics[width=0.86\textwidth]{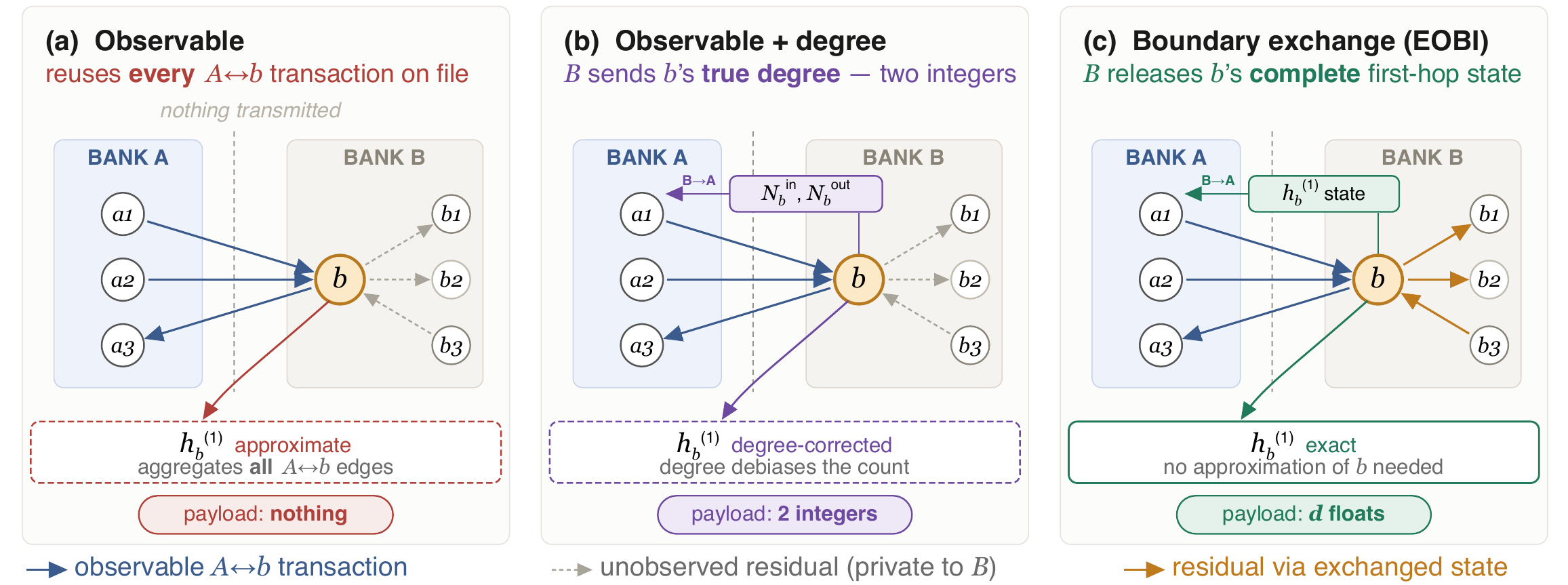}
\caption{Three prices for cross-bank context at a frontier account $b$: \obs{} sends no payload, \obsd{} two degree counts, and \eobi{} the owner-computed cut state (first-hop state shown; the benchmark record also carries $h_b^{(2)}$ when consumed by the readout, i.e., $2d$ floats); raw residual transactions remain at $B$.}
\label{fig:interfaces}
\end{figure*}

\textbf{Payload accounting.}
Figure~\ref{fig:interfaces} summarizes the three payloads. For each owner $B$ and observer $A$, let
\[
\partial_{B\to A}=\{b:o(b)=B,\ \mathcal T_A(b)\ne\varnothing\},
\quad F=\sum_B\sum_{A\ne B}|\partial_{B\to A}|.
\]
\obs{} sends no account-specific payload beyond model updates. \obsd{} sends two counts per directional frontier registration at each of $S=3$ cumulative snapshots aligned with the train/validation/test boundaries. Counts use only transactions available up to the relevant snapshot; they do not use later transactions. The observer log-transforms and standardizes the vector before concatenating it to the remote initial representation. A two-state-level \eobi{} record of width $d$ contains $2d$ scalar values per registration per refresh. Over $R$ rounds, its scalar-value count is $dR/S$ times the degree interface's count; a byte comparison additionally depends on integer and floating-point encodings. Neither raw residual transactions nor a pooled ledger are transferred. Payload counts network traffic only: \obs{} still rebuilds and caches one approximation per frontier registration at each refresh, and \eobi{} additionally requires the intermediate synchronizations of Section~\ref{sec:context}.

\section{Interface Value by Typology}
\label{app:typology}

\begin{figure}[t]
\centering
\includegraphics[width=\columnwidth]{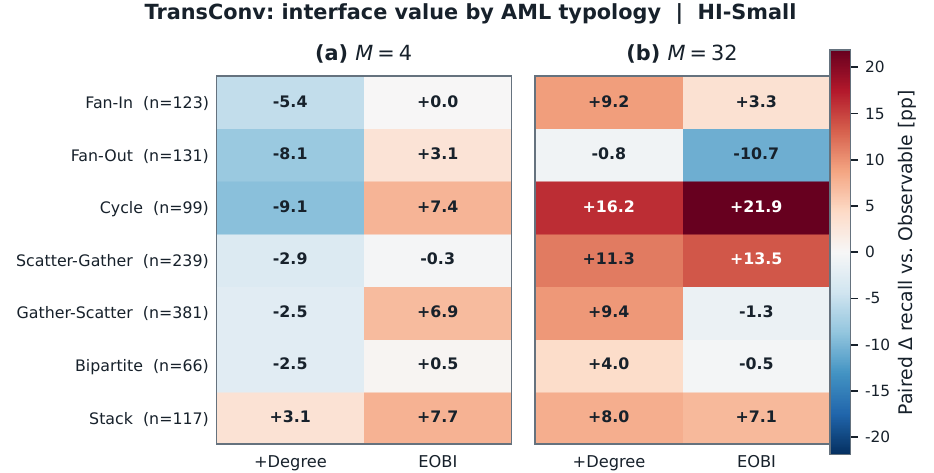}
\caption{Per-typology interface value on HI-Small (TransConv) at
$M{=}4$ and $M{=}32$: paired recall difference vs.\ \obs{}, in
percentage points; $n$ counts typology-annotated test positives.}
\label{fig:typology}
\end{figure}

The benchmark results of Section~\ref{sec:benchmark} pool all laundering typologies. Figure~\ref{fig:typology} disaggregates them by the patterns of Figure~\ref{fig:patterns}, using the typology annotations of AMLworld \cite{altman2023realistic}. For TransConv, the backbone of the selective controlled study (Appendix~\ref{app:synthetic}), each cell reports the recall change when \obsd{} or \eobi{} is added to \obs{}, at low ($M{=}4$) and high ($M{=}32$) fragmentation. Three profiles emerge, matching the residual classes of Section~\ref{sec:privacy}, which is the closest the benchmark comes to a real-data counterpart of the controlled study.

\textbf{Cycle.} As predicted for a selective residual, the transaction closing a cycle may lie outside the labeling silo, depriving the observable slice of the required witness. Consistently, \eobi{} already helps at $M{=}4$, and the largest jump appears at $M{=}32$, where deeper partitioning disperses more witnesses into the residual. The degree also gains there, through a different channel: it cannot identify a witness, but it corrects the attention normalization that the slice biases.

\textbf{Fan-In.} The cardinal pattern par excellence. The degree is useless or even harmful at low fragmentation, and strongly beneficial at high fragmentation, where it repairs the count that the slice systematically underestimates. Two integers outperform the $d$-float \eobi{} payload: for a count, the sufficient statistic is the count, and softmax normalization divides out the very multiplicity that the degree transmits exactly.

\textbf{Fan-Out.} The clear counterexample. The hub sends every transaction of the pattern, so the labeling silo observes it in full by construction and residual context carries no label information; as $M$ grows it only injects more uninformative remote states into the hub, and \eobi{} accordingly decays.

\textbf{Gather--Scatter and Scatter--Gather.} These mix the two hub roles and inherit intermediate profiles: at high fragmentation the former is repaired chiefly by the degree, while the latter benefits from both payloads, consistent with its mixed cardinal and selective structure.

Taken together, the typology profiles indicate that the residual classes of the controlled study are not an artifact of its generator: on HI-Small, cardinal patterns are repaired by the count, selective patterns by the exchanged state, and patterns whose evidence is locally complete by neither.

\section{Controlled Residual-Context Study}
\label{app:synthetic}
The controlled study independently varies observable coverage and residual signal type. Its privacy--utility results are separate from the unprotected HI-Small benchmark.

\subsection{Generator, Interfaces, and Evaluation}
\label{app:generator}
Every example contains an always-observable target transaction and 16
remote boundary accounts, each with at most $N_{\max}=80$ events.
Labels are assigned by a generator with full visibility. It first
computes two scores: a target-edge score $s_e$ from the target
transaction's attributes, and a residual-context score $s_c$ from the
complete activity of the boundary accounts. The two are then blended
into
\begin{equation}
s_\lambda =
\frac{\lambda s_e + (1-\lambda)s_c}
{\sqrt{\lambda^2+(1-\lambda)^2}}+\epsilon,
\qquad
y=\mathbf{1}[s_\lambda\geq q_{0.8}],
\label{app:eq:label}
\end{equation}
where $\lambda\in\{0,.25,.5,.75,1\}$ controls how much of the label
the target transaction alone explains: at $\lambda=1$ residual context
is irrelevant by construction, and smaller values make it increasingly
informative. Gaussian noise $\epsilon$ sets task difficulty, and
thresholding fixes 20\% prevalence. The model never sees these scores:
it must predict $y$ from whatever its interface reveals of the
boundary activity.

Observable coverage is controlled through
$\rho\in\{.1,.3,.6,1\}$: for
each boundary account, \obs{} estimates the first-hop key and value
from a random $\rho$ fraction of its events, with at least one event
always retained. The degree interface augments this partial view with
the account's true total event count $N_b$, whereas \eobi{} provides
the exact first-hop state computed from all events. Here the count is
one sufficient statistic for the synthetic cardinal regime; the
HI-Small interface instead supplies the two directed in/out counts.

We generate a separate synthetic dataset for each residual class. In
the \emph{distributional} regime, $s_c$ depends on the mean activity
profile, which can be estimated from a representative observable
sample. In the \emph{cardinal} regime, it depends on the total event
count, which cannot in general be recovered exactly from an incomplete
slice but is supplied by the degree interface. In the
\emph{selective} regime, only one of the 16 boundary accounts
matters: the target transaction carries a query vector and $s_c$ is
computed from the single account whose key matches it, while the
others are high-activity distractors. A count does not identify that
account, and the relevant evidence may fall outside the observable
slice. This is the synthetic analog of a cycle-closing witness hidden
from the observing institution.

We pair these regimes with GIN, PNA, and TransConv poolers, which use
mean/sum aggregation, multiple statistics, and target-conditioned
attention, respectively---matched to the mechanism each class calls
for. For each configuration, we generate 3{,}000 examples per seed
with a 60/20/20 train/validation/test split, share initialization and
minibatch order across interfaces, and report mean PR-AUC over three
dataset seeds.

\begin{table}[t]
\caption{Unprotected residual value at low coverage ($\rho=.1$, $\lambda=.75$): three-seed mean PR-AUC. Here +count releases one total count.}
\label{tab:clean}
\centering
\begin{tabular}{llccc}
\toprule
Residual & Pooler & \obs{} & +count & \eobi{} \\
\midrule
Distributional & GIN & $.853$ & $.854$ & $.855$ \\
Cardinal & PNA & $.763$ & $.796$ & $.806$ \\
Selective & TransConv & $.822$ & $.824$ & $.869$ \\
\bottomrule
\end{tabular}
\end{table}

\subsection{Unprotected Interface Value}
\label{app:clean-utility}
Table~\ref{tab:clean} reports the low-coverage stress slice. The full grid below varies both target-edge information and observable coverage.

\begin{figure}[t]
\centering
\includegraphics[width=0.95\columnwidth]{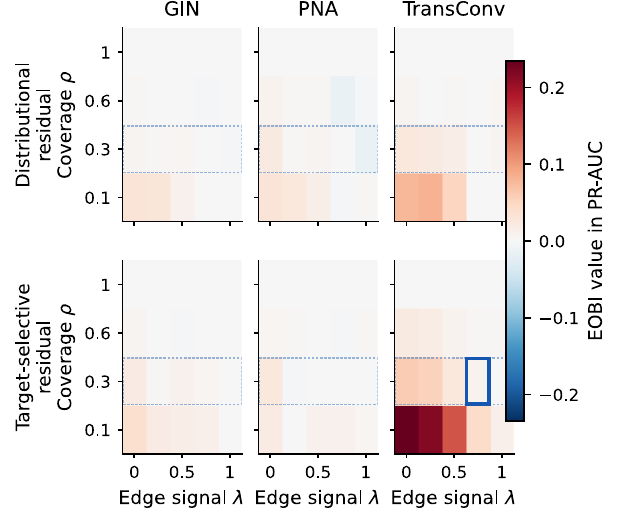}
\caption{Controlled value of unprotected boundary-state exchange over
the $(\lambda,\rho)$ grid. Each cell reports the three-seed mean
$\Delta$PR-AUC (\eobi{}$-$\obs{}); columns vary the target-edge
contribution $\lambda$, rows the observable coverage $\rho$. Dashed
outlines mark $\rho=.3$, within the HI-Small coverage range
$.245$--$.366$; the solid box marks the TransConv $\lambda=.75$
reference slice. Matching coverage alone does not equate a synthetic
task with the benchmark.}
\Description{Heatmaps show the incremental PR-AUC of boundary-state exchange over Observable for three synthetic residual classes as observable coverage and the target-edge label contribution vary. Gains concentrate at lower coverage when residual context remains relevant.}
\label{app:fig:heatmap}
\end{figure}

Figure~\ref{app:fig:heatmap} maps the \eobi{}$-$\obs{} difference over
the full $(\lambda,\rho)$ grid. Two limiting cases are sanity checks:
at $\rho=1$ the interfaces receive the same information, and at
$\lambda=1$ residual context is irrelevant to the label. The
informative region lies between these extremes: gains concentrate
where coverage is low and the residual label-relevant, and shrink
quickly as coverage grows. GIN and PNA also gain when the residual is
label-relevant, so exchange is not an attention-only phenomenon.
These gains concern unprotected payloads; the following study asks
whether they survive formal protection.

\section{Privacy Mechanisms, Accounting, and Utility}
\label{app:privacy}
\subsection{Scope and Threat Model}
\label{app:privacy-scope}
Within the scope of this analysis, \obs{} introduces no
account-specific payload beyond the model updates. We therefore ask
how much formal protection an additional residual payload can absorb
before its value falls below the clean observable interface. We reuse
the datasets and pairing of Appendix~\ref{app:synthetic}, fixing the
selective TransConv setting at the reference slice
($\lambda=.75$, $\rho=.1$), with the distributional regime as negative
control. Only the boundary states are protected; the target edge is
observable by definition.

The owner is trusted to compute the mechanism correctly; the
adversary is an honest-but-curious recipient who knows the mechanism
and holds auxiliary information. Neighboring inputs keep the
recipient's observable slice and target transaction fixed and alter
only residual activity. Each study analyzes a single protected
payload release. The state-sensitivity account treats the encoder
parameters as fixed: it does not account for changes to a learned
encoder or for disclosure through model updates. Retraining the
recipient on a fixed noisy release measures its utility, not the
privacy of the preceding federated training process. In particular,
the Gaussian analysis below concerns synthetic first-hop states,
not the complete multi-level cut-state record used by the benchmark.

\subsection{Protecting the State Payload}
\label{app:state-mechanism}
For one example, let $H$ collect one owner-computed first-hop state
per boundary account. We clip each row to $\ell_2$ norm at most $C=3$
and release once
\begin{equation}
\mathcal M(H)=\operatorname{clip}_C(H)+
\mathcal N\!\bigl(0,(\sigma C)^2I\bigr),
\label{app:eq:gaussian}
\end{equation}
where larger $\sigma$ provides stronger protection but greater
payload distortion. We account for each $\sigma$ under two guarantees.

The first guarantee is event-level differential privacy: it bounds
what the released payload can reveal about any single remote event.
Call two inputs adjacent when they differ by replacement of one
remote event of one boundary account. In this construction, with the
encoder fixed, such a change affects one row of $H$, so the clipped
matrices differ by at most $2C$ in Euclidean norm. The analytic
Gaussian analysis~\cite{balle2018analyticGaussian} then gives, at
$\delta=10^{-5}$, an $\epsilon_e$ such that for every adjacent pair
$H,H'$ and output set $S$,
\begin{equation}
\Pr[\mathcal M(H)\in S]\le
e^{\epsilon_e}\Pr[\mathcal M(H')\in S]+\delta.
\label{app:eq:event-dp}
\end{equation}
Because the bound must hold for the most damaging substitution, this
worst-case account can be much looser than an account restricted to
smaller state changes.

The second account is a distance-scoped profile of the same Gaussian
mechanism, motivated by metric privacy~\cite{chatzikokolakis2013metrics}.
It evaluates the preceding indistinguishability inequality with
sensitivity $r$, restricting it to pairs satisfying
$\|\operatorname{clip}_C(H)-\operatorname{clip}_C(H')\|_2\le r$.
To choose $r$, we construct 1{,}536 event-replacement pairs per regime,
each replacing the most influential event in a held-out example with
a less influential one, and set $r$ to the 90th percentile $r_{90}$
of the resulting clipped-state distances. We denote the resulting
account by $\epsilon_m$.
Within $r_{90}$ the guarantee is formal; the 90\% coverage is an
empirical property of the constructed pairs, not a worst-case
statement or a guarantee covering 90\% of arbitrary future events.
$\epsilon_e$ and $\epsilon_m$ answer different questions and are not
interchangeable, nor directly comparable to guarantees under other
adjacency relations~\cite{guo2026cefedgnn}.

\subsection{Paired Utility and State Results}
\label{app:state-results}
For each noise level $\sigma$, we retrain the recipient model on five
independent noise draws per dataset, each a fixed noisy release,
sharing initialization and minibatch order with the observable
baseline. We report
$\Delta=\mathrm{PR\mbox{-}AUC}_{\mathrm{private\text{-}EOBI}}
-\mathrm{PR\mbox{-}AUC}_{\mathrm{observable}}$
with paired 95\% hierarchical-bootstrap intervals over dataset seeds
and noise draws (5{,}000 resamples). Exchange is called useful
(observable preferred) only when the entire interval lies above
(below) zero.

Clipping alone preserves the selective gain ($+.044$); exchange is
confirmed useful up to $\sigma=.075$ ($\epsilon_m=3.55$,
$\epsilon_e=468$), and observable is confirmed preferable from
$\sigma=.175$ ($\epsilon_m=1.37$, $\epsilon_e=113$). At the noise
levels in between, the intervals contain zero and neither direction
is confirmed. Within the tested grid, exchange thus stops beating
the clean observable slice while the worst-case $\epsilon_e$ remains
very large. This conclusion concerns the tested mechanism and tasks,
not an impossibility result for all private state-exchange methods.

Composition further increases the cost of repeated disclosure. Under
the same per-release distance and noise account, the
$r_{90}$-scoped bound at $\sigma=.075$ rises from $3.55$ for one
release to $74.7$ for 100~\cite{mironov2017renyi}. This accounting
illustration is not a privacy guarantee for the full 100-round
benchmark training procedure.

We also repeat the study at $\lambda=.5$, where unprotected \eobi{}
has clear value for TransConv. The clean gain rises to $+.147$, yet
\obs{} is already confirmed preferable at $\sigma=.125$
($\Delta=-.020$ [$-.032,-.007$]) while the worst-case account remains
loose ($\epsilon_e\approx195$). The distributional control shows
clipping-only $\Delta=+.001$ [$-.007,+.009$] and no confirmed useful
noisy point: noise does not manufacture residual value.

\subsection{Protecting the Total-Count Payload}
\label{app:degree-mechanism}
To complete the privacy--utility analysis, we ask the same question
of the degree. The experiment moves to the cardinal regime, where
the count carries value in the clear (PNA; $\rho=.1$;
$\lambda=.5$; datasets, pairing, and bootstrap unchanged), and
privatizes each frontier account's total count $N_b$, the sufficient
statistic in this regime. At each snapshot, the owner releases one
noisy count $N_b+Z$ per account, with $Z$ two-sided geometric,
\begin{equation}
\Pr[Z=z]\propto\exp(-\epsilon|z|/2).
\label{app:eq:geometric}
\end{equation}
The same noisy value is reused for every observer, rather than
drawing new noise for each recipient. Each observer clips it to the
feasible range $[n_{A,b},N_{\max}]$ as
post-processing~\cite{ghosh2009universally}. The bound $N_{\max}$ and
the observer's already-recorded count $n_{A,b}$ are fixed in the
neighboring-input comparison; no extra privacy loss is charged for
this post-processing of the single release.

Adjacency adds or removes one residual transaction, changing at most
two account counts by one. The released count vector therefore has
$\ell_1$ sensitivity at most 2, so a single snapshot release with
the stated geometric mechanism satisfies pure $\epsilon$-DP. The
three snapshot releases compose to $3\epsilon$. Reusing a release
across observers does not create independent additional releases.

Here the gain survives: the clean $+.047$ over \obs{} is nearly
intact at $\epsilon=.1$ ($+.042$) and, though shrinking, remains
confirmed positive even at $\epsilon=.01$
($+.004$ [$+.002,+.006$], count MAE $32.2$ on $N_{\max}=80$).
The contrast with state exchange is qualitative: the state stops
paying while its worst-case account is still loose, whereas the
total count retains utility under strong formal protection in this
controlled cardinal study. The two privacy parameters protect
different objects under different adjacency relations and are not
directly comparable. These results do not measure protected in/out
degree utility on HI-Small, and neither mechanism accounts for
privacy leakage through federated model updates.

\end{document}